\documentclass[11pt]{article}

\usepackage{times}
\usepackage{fullpage}
\usepackage{amsfonts,amssymb,amsmath,xspace}
\usepackage{latexsym}
\usepackage{url}
\usepackage{verbatim}
\usepackage[pdftex]{color}
\usepackage[breaklinks]{hyperref}

\newcommand{\R}{\mathbb{R}}
\newcommand{\eps}{\varepsilon}

\newcommand{\floor}[1]{\lfloor #1 \rfloor}
\newcommand{\ceil}[1]{\left\lceil #1 \right\rceil}

\newcommand{\braket}[2]{\langle #1|#2 \rangle}

\newcommand{\ket}[1]{| #1 \rangle}

\newcommand{\fro}[1]{\|#1\|_F}
\newcommand{\trn}[1]{\|#1\|_{tr}}

\def\01{\{0,1\}}

\newcommand{\Tr}{\mathrm{Tr}}

\newcommand{\DISJ}{\mathrm{DISJ}}

\newcommand{\IP}{\mathrm{IP}}
\newcommand{\diag}{\mathrm{diag}}

\newcommand{\ignore}[1]{}

\newtheorem{theorem}{Theorem}
\newtheorem{lemma}[theorem]{Lemma}
\newtheorem{corollary}[theorem]{Corollary}
\newtheorem{proposition}[theorem]{Proposition}
\newtheorem{definition}[theorem]{Definition}

\newtheorem{example}[theorem]{Example}

\newtheorem{fact}[theorem]{Fact}

\newcommand{\thmref}[1]{\hyperref[#1]{{Theorem~\ref*{#1}}}}
\newcommand{\lemref}[1]{\hyperref[#1]{{Lemma~\ref*{#1}}}}
\newcommand{\corref}[1]{\hyperref[#1]{{Corollary~\ref*{#1}}}}
\newcommand{\eqnref}[1]{\hyperref[#1]{{Equation~(\ref*{#1})}}}
\newcommand{\claimref}[1]{\hyperref[#1]{{Claim~\ref*{#1}}}}
\newcommand{\remarkref}[1]{\hyperref[#1]{{Remark~\ref*{#1}}}}
\newcommand{\propref}[1]{\hyperref[#1]{{Proposition~\ref*{#1}}}}
\newcommand{\factref}[1]{\hyperref[#1]{{Fact~\ref*{#1}}}}
\newcommand{\defref}[1]{\hyperref[#1]{{Definition~\ref*{#1}}}}
\newcommand{\exampleref}[1]{\hyperref[#1]{{Example~\ref*{#1}}}}
\newcommand{\hypref}[1]{\hyperref[#1]{{Hypothesis~\ref*{#1}}}}
\newcommand{\secref}[1]{\hyperref[#1]{{Section~\ref*{#1}}}}
\newcommand{\chapref}[1]{\hyperref[#1]{{Chapter~\ref*{#1}}}}
\newcommand{\apref}[1]{\hyperref[#1]{{Appendix~\ref*{#1}}}}
\newcommand\rank{\mbox{\tt {rank}}\xspace}
\newcommand\prank{\mbox{\tt {rank}$_{\tt psd}$}\xspace}
\newcommand\rprank{\mbox{\tt {rank}$^{\R}_{\tt psd}$}\xspace}
\newcommand\alice{\mbox{\sf Alice}\xspace}
\newcommand\bob{\mbox{\sf Bob}\xspace}

\newenvironment{proof}[1][Proof: ]
{\noindent {\bf #1}}
{{\hfill $\Box$}\\
 \smallskip}

\begin{document}

\title{Some upper and lower bounds on PSD-rank}
\author{Troy Lee\thanks{School of Physics and Mathematical Sciences, Nanyang Technological University and Centre for Quantum Technologies, Singapore. Email:troyjlee@gmail.com} \\
\and
Zhaohui Wei\thanks{School of Physics and Mathematical Sciences, Nanyang Technological University and Centre for Quantum Technologies, Singapore. Email:weizhaohui@gmail.com} \\
\and Ronald de Wolf\thanks{CWI and University of Amsterdam,
Amsterdam, The Netherlands. Email:rdewolf@cwi.nl} }
\date{}
\maketitle

\begin{abstract}
Positive semidefinite rank (PSD-rank) is a relatively new quantity
with applications to combinatorial optimization and communication
complexity. We first study several basic properties of PSD-rank, and
then develop new techniques for showing lower bounds on the
PSD-rank.  All of these bounds are based on viewing a positive
semidefinite factorization of a matrix $M$ as a quantum
communication protocol. These lower bounds depend on the entries of
the matrix and not only on its support (the zero/nonzero pattern),
overcoming a limitation of some previous techniques. We compare
these new lower bounds with known bounds, and give examples where
the new ones are better. As an application we determine the PSD-rank
of (approximations of) some common matrices.
\end{abstract}

\section{Introduction}

\subsection{Background}

We study the properties of \emph{positive semidefinite factorizations}.
Such a factorization (of size $r$) of a nonnegative $m$-by-$n$
matrix $A$ is given by $r$-by-$r$ positive semidefinite
matrices $E_1, \ldots, E_m$ and $F_1,\ldots, F_n$ satisfying
$A(i,j)=\Tr(E_i F_j)$.  The \emph{positive semidefinite rank} (PSD-rank) of
$A$ is the smallest $r$ such that $A$ has a positive semidefinite
factorization of size~$r$. We denote it by $\prank(A)$.  The notion of PSD-rank has been
introduced relatively recently because of applications to
combinatorial optimization and communication complexity
\cite{GPT11,FMP+12}. These applications closely parallel those of the
\emph{nonnegative rank} of~$A$, which is the minimal number of rank-one nonnegative matrices that sum to~$A$.

In the context of combinatorial optimization, a polytope $P$ is associated with a nonnegative matrix known as
the \emph{slack matrix} of $P$.  A classic result of Yannakakis~\cite{Yannakakis91} shows that the nonnegative rank
of the slack matrix of $P$ characterizes the size of a natural way of formulating the optimization of a linear function
over $P$ as a linear program.  More precisely, the nonnegative rank of the slack matrix of~$P$ equals the
\emph{linear extended formulation} size of $P$, which is the minimum number of facets of a (higher-dimensional) polytope~$Q$ that projects to~$P$.
Analogously, the PSD-rank of the slack matrix of $P$ captures the size of a natural way of
optimizing a linear function over $P$ as a \emph{semidefinite} program \cite{GPT11,FMP+12}.  More precisely, the PSD-rank
of the slack matrix of $P$ is equal to the \emph{positive semidefinite extension} size of $P$, which is the smallest
$r$ for which $P$ can be expressed as the projection of an affine slice of the cone of
$r$-dimensional positive semidefinite matrices.

There have recently been great strides in understanding linear extended formulations, showing that the linear
extended formulation size for the traveling salesman and matching polytopes is exponentially large in the number of vertices of the underlying graph~\cite{FMP+12,Rot14}.
It is similarly conjectured that the traveling salesman polytope requires superpolynomial \emph{positive semidefinite
extension complexity}, and proving this requires showing lower bounds on the PSD-rank of the corresponding slack matrix.

In communication complexity, nonnegative and PSD-rank arise in the model of computing a function
$f : \01^m \times \01^n \rightarrow \R_+$ in expectation.  In this model, \alice has an input $x \in
\01^m$, \bob has an input $y \in \01^n$ and their goal is to communicate in order for \bob to output a nonnegative
random variable whose expectation is $f(x,y)$.  The associated communication matrix for this problem is a
$2^m$-by-$2^n$ matrix whose $(x,y)$ entry is $f(x,y)$.  The nonnegative rank of the communication matrix of $f$
characterizes the amount of classical communication needed to compute $f$ in expectation \cite{CFF+12}.
Analogously, the PSD-rank of the communication matrix of $f$ characterizes the amount of \emph{quantum}
communication needed to compute $f$ in expectation~\cite{FMP+12}.
Alternatively, one can consider the problem where
\alice and \bob wish to generate a probability distribution $P(x,y)$
using shared randomness or shared entanglement, but without
communication.  The number of bits of shared randomness or qubits of
shared entanglement are again characterized by the nonnegative rank
and PSD-rank, respectively \cite{Zha12,JSWZ12}.

Accordingly, providing lower and upper bounds on the PSD-rank is interesting in the context of communication complexity
as well.  Here we will pin down, up to constant factors, the PSD-rank of some common matrices
studied in communication complexity like inner product and non-equality.

\subsection{Our results}

As PSD-rank is a relatively new quantity, even some basic questions about its
behavior remain unanswered.  We address several properties here.
First we show that, unlike the usual rank, PSD-rank is not strictly
multiplicative under tensor product: we give an example of a matrix
$P$ where $\prank(P \otimes P) < \prank(P)^2$.  We do this by making a connection between PSD-rank and
planar geometry to give a simple sufficient condition for when the PSD-rank is not full.

The second question we address is the dependence of PSD-rank on the underlying field.  At the Dagstuhl
Seminar 13082 (February 2013), Dirk Oliver Theis raised the question if the PSD-rank where the factorization is by \emph{real} symmetric
PSD-matrices is the same as that by \emph{complex} Hermitian PSD-matrices.  It is easy to see that the real PSD-rank
can be at most a factor of~2 larger than the complex PSD-rank; we give an infinite family of matrices where the
real PSD-rank is asymptotically a factor of~$\sqrt{2}$ larger than the complex PSD-rank.

Our main goal in this paper is showing lower bounds on the
PSD-rank, a task of great importance to both the applications to
combinatorial optimization and communication complexity mentioned
above. Unfortunately, at this point very few techniques exist to
lower bound the PSD-rank.

One lower bound direction is to consider
only the \emph{support} of the matrix, that is the pattern of zero/nonzero entries.  For the
nonnegative rank, this method can show good lower bounds---in
particular, support-based arguments sufficed to show exponential
lower bounds on the linear extension complexity of the traveling
salesman polytope \cite{FMP+12}.
For the PSD-rank, however, support-based arguments cannot show lower
bounds larger than the rank of the matrix \cite{LeeTheis12}.  This
means that for cases like the traveling salesman polytope, where we
believe the positive semidefinite extension complexity is
superpolynomial in the rank of the slack matrix, other techniques
need to be developed.

We develop three easy-to-compute lower bounds on PSD-rank.  All
three depend on the values of the matrix and not only on its support
structure---in particular, they can show nontrivial lower bounds for
matrices without zero entries.  All three are derived from the
viewpoint of PSD-rank of a nonnegative matrix as a quantum
communication protocol. We compare these lower bounds with previous
techniques and show examples where they are better.

We also give nearly tight bounds on the PSD-rank of (approximations of) the identity matrix and on the PSD-rank of the matrix corresponding to the inner product and nonequality functions.

\section{Preliminaries}

Let $M=[M(i,j)]$ be an arbitrary $m$-by-$n$ matrix of rank~$r$ with
the $(i,j)$-th entry being $M(i,j)$, and let
$\sigma_1,\sigma_2,\ldots,\sigma_r$ be the nonzero singular values
of~$M$. The \emph{trace norm} of~$M$ is defined as $\trn{M}=\sum_i
\sigma_i$, and the \emph{Frobenius norm} of $M$ is defined as
$\fro{M}=(\sum_i \sigma_i^2)^{1/2}$; this equals $(\sum_{i,j}
M(i,j)^2)^{1/2}$. Note that $\fro{M}\leq \trn{M}$. By the
Cauchy-Schwarz inequality we have
\begin{equation}
\label{eq:trace_norm_bound}
\rank(M) \ge \left(\frac{\trn{M}}{\fro{M}}\right)^2
\end{equation}

\subsection{PSD-rank}
Since it is the central topic of this paper, we repeat the definition of PSD-rank from the introduction:

\begin{definition}
Let $A$ be a nonnegative $m$-by-$n$ matrix.
A  \emph{positive semidefinite factorization} of size $r$ of~$A$
is given by $r$-by-$r$ positive semidefinite matrices $E_1, \ldots, E_m$ and $F_1,\ldots, F_n$ satisfying $A(i,j)=\Tr(E_i F_j)$.
The \emph{positive semidefinite rank} (PSD-rank, $\prank(A)$) of
$A$ is the smallest integer~$r$ such that $A$ has a positive semidefinite
factorization of size~$r$.
\end{definition}

Note that for a nonnegative matrix $A$, the PSD-rank is unchanged
when we remove all-zero rows and columns.  Also, for
nonnegative diagonal matrices $D_1, D_2$, the PSD-rank of $D_1 A D_2$
is at most that of $A$. Throughout this paper we will use these
facts to achieve a particular normalization for~$A$.  In particular,
we will frequently assume without loss of generality that each
column of $A$ sums to one, i.e., that $A$ is a stochastic matrix.

The following lemma is very useful for giving upper bounds on the PSD-rank.

\begin{lemma}(\cite{Zha12})\label{lem:1D}
If $A$ is a nonnegative matrix, then
\[
\prank(A) \leq \min_{M:\ M\circ \bar M = A} \rank(M),
\]
where $\circ$ is the Hadamard product (entry-wise product) and $\bar
M$ is the entry-wise complex conjugate of $M$.
\end{lemma}

In the definition of PSD-rank, we allow the matrices of the
PSD-factorization to be arbitrary Hermitian PSD matrices, with complex-valued entries.
One can also consider the \emph{real} PSD-rank, where the matrices of the factorization are
restricted to be real symmetric PSD matrices.  For a nonnegative matrix $A$,
we denote its real PSD-rank by $\rprank(A)$.

We now review some existing lower bound methods for the PSD-rank.
Firstly, it is well known that the PSD-rank cannot be much smaller than the normal rank $\rank(A)$ of~$A$.

\begin{definition}
For a nonnegative matrix~$A$, define
$$
B_1(A)=\sqrt{\rank(A)}\mbox{ and }
B'_1(A)=\frac{1}{2}\left(\sqrt{1+8\rank(A)}-1\right).
$$
\end{definition}

\begin{fact}(\cite{GPT11})\label{ft:trivial}
$\prank(A)\geq B_1(A)$ and $\rprank(A)\geq B'_1(A)$.
\end{fact}

This bound does not look very powerful since, as stated in
the introduction, usually our goal is to show lower bounds on the
PSD-rank that are superpolynomial in the rank. Surprisingly,
however, this bound can be nearly tight and we give two examples in
\secref{sec:B1tight} where this is the case.

Jain et al.~\cite{JSWZ12} proved that quantum communication
needed for two separated players to generate a joint probability
distribution $P$ is completely characterized by the logarithm of the
PSD-rank of $P$. Combining this result and Holevo's bound, a trivial
lower bound for PSD-rank is given by mutual information.
\begin{definition}
Let $P = [P(i,j)]_{i,j}$ be a two-dimensional probability
distribution between two players $A$ and $B$. Define $B_2(P)=2^{H(A:B)}$,
where $H(A:B)$ is the mutual information between the two players.
\end{definition}

\begin{fact}\label{ft:mutual}
$\prank(P)\geq B_2(P)$.
\end{fact}

As an application of this lower bound, it is easy to see that the
PSD-rank of a diagonal nonnegative matrix is the same as its normal rank.

The only result we are aware of showing lower bounds
on PSD-rank asymptotically larger than the rank is a very general result
of Gouveia et al.~\cite{GPT11} that shows the following.
\begin{fact}(\cite{GPT11})
Let $P \subseteq \R^d$ be a polytope with $f$ facets and let $S_P$
be its associated slack matrix. Let $T=\sqrt{\log(f)/d}$.  Then
\[
\prank(S_P) =\Omega \left (\frac{T}{\sqrt{\log(T)}}\right )
\]
\end{fact}

In particular, this shows that the slack matrix of a regular $n$-gon
in $\R^2$, which has $n$ facets and rank~3, has PSD-rank $\Omega(\sqrt{\log n/\log \log n})$.
The nonnegative rank of this matrix is known to be $\Theta(\log n)$~\cite{BentalNemirovski01}.

\subsection{Quantum background}\label{ssecquantum}
A \emph{quantum state} $\rho$ is a positive semidefinite matrix with trace $\Tr(\rho)=1$.
A \emph{POVM} (``Positive Operator Valued Measure'') ${\cal E}=\{E_m\}$ consists of positive semidefinite matrices $E_m$ that sum to the identity.
When measuring a quantum state $\rho$ with this POVM, the outcome is $m$ with probability $p_m=\Tr(\rho E_m)$.

For our purposes, a \emph{(one-way) quantum protocol} between two players Alice (with input~$x$) and Bob (with input~$y$) is the following: Alice sends a quantum state $\rho_x$ to Bob, who measures it with a POVM ${\cal E}_y=\{E_m\}$. Each outcome $m$ of this POVM is associated with a nonnegative value, which is Bob's output.
We say the protocol \emph{computes an $m$-by-$n$ matrix~$M$ in expectation} if, for every $x\in[m]$ and $y\in[n]$, the expected value of Bob's output equals~$M(x,y)$. Fiorini et al.~\cite{FMP+12} showed that the minimal dimension of the states $\rho_x$ in such a protocol is either $\prank(M)$ or $\prank(M)+1$, so the minimal number of qubits of communication is essentially $\log \prank(M)$.

For two quantum states $\rho$ and $\sigma$, we definite the \emph{fidelity} between them by
\[
F(\rho,\sigma)=\trn{\sqrt{\sigma}\sqrt{\rho}}.
\]
See \cite[Chapter~9]{NC00} for additional properties and equivalent formulations of the fidelity.
The fidelity between two probability distributions $p,q$ is $F(\diag(p), \diag(q))$.

The following two facts about fidelity will be useful for us.
\begin{fact}\label{ft:fidelity}
If $\sigma,\rho$ are quantum states, then $\Tr(\sigma \rho) \le F(\sigma,\rho)^2$.
\end{fact}

\begin{proof}
We have $\Tr(\sigma \rho)=\Tr((\sqrt{\sigma}\sqrt{\rho})(\sqrt{\sigma}\sqrt{\rho})^\dag) = \fro{\sqrt{\sigma}\sqrt{\rho}}^2 \le \trn{\sqrt{\sigma}\sqrt{\rho}}^2=F(\sigma,\rho)^2$.
\end{proof}

\begin{fact}(\cite{NC00})\label{ft:fidelClasQuan}
If $\sigma,\rho$ are quantum states, then
\[
F(\sigma,\rho)=\min_{\{E_m\}}F(p,q),
\]
where the minimum is over all POVMs $\{E_m\}$, and $p$ and $q$ are
the probability distributions when $\rho$ and $\sigma$ are measured
by POVM $\{E_m\}$ respectively, i.e., $p_m=\Tr(\rho E_m)$, and
$q_m=\Tr(\sigma E_m)$ for any $m$.
\end{fact}

\section{Some properties of PSD-rank}
The PSD-rank is a relatively new quantity, and even some of its
basic properties are still not yet known.  In this section we give a simple
condition for the PSD-rank of a matrix to not be full.
We then use this condition to show that PSD-rank can be strictly
sub-multiplicative under tensor product.  Finally, we investigate
the power of using complex Hermitian over real symmetric
matrices in a PSD factorization.

\subsection{A sufficient condition for PSD-rank to be less than maximal}
We first need a definition and a simple lemma.
Let $v \in \R^m$ be a vector.  We say that an entry $v_k$ is \emph{dominant}
if $|v_k| > \sum_{j \ne k} |v_j|$.
\begin{lemma}
\label{lem:tri}
Suppose that $v \in \R^m$ is nonnegative and has no dominant entries.  Then there exist complex units $e^{i\theta_j}$
such that $\sum_j v_j e^{i \theta_j} =0$.
\end{lemma}

\begin{proof}
Let $v \in \R^m$.  If $m=1$ then $v$ has a dominant entry and there is nothing to prove.
If $m=2$ and $v$ has no dominant entries, then $v_1 = v_2$ and the lemma holds as $v_1 - v_2 =0$.

The first interesting case is $m=3$.  That $v$ has no dominant entries means there is a triangle with side lengths
$v_1, v_2, v_3$, as these satisfy the triangle inequality with respect to all permutations.  Letting
$v_1 e^{\rm{i}\theta_1}, v_2 e^{\rm{i}\theta_2}, v_3 e^{\rm{i}\theta_3}$ be the vectors (oriented head to tail) defining the sides of this triangle gives
$v_1 e^{\rm{i}\theta_1} + v_2 e^{\rm{i}\theta_2}+v_3 e^{\rm{i}\theta_3}=0$ as desired.

We can reduce the case $m>3$ to the case $m=3$.  Without loss of generality, order $v$ such that
$v_1 \ge v_2 \ge \cdots \ge v_m$.  Choose $k$ such that
\[
\sum_{j=k+1}^m v_j \le \sum_{j=2}^k v_j \le \sum_{j=k+1}^m v_j + v_1.
\]
Then $v_1, \sum_{j=2}^k v_j, \sum_{j=k+1}^m v_j$ mutually satisfy the triangle inequality and we can repeat the
construction from the case $m=3$ with these lengths.
\end{proof}

Using the construction of \lemref{lem:1D}, we can give a simple condition for $A$ not to have full PSD-rank.

\begin{theorem}
\label{thm:not_full}
Let $A$ be an $m$-by-$n$ nonnegative matrix, and $A'$ be the
entry-wise square root of $A$ (so $A'$ is nonnegative as well). If every
column of $A'$ has no dominant entry, then the PSD-rank of $A$ is less than~$m$.
\end{theorem}
\begin{proof}
As each column of $A'$ has no dominant entry, by \lemref{lem:tri} there exist
complex units $e^{i\theta_{jk}}$ such that $\sum_j A'(j,k) e^{i
\theta_{jk}} =0$ for every $k$.  Define $M(j,k)=A'(j,k) e^{i
\theta_{jk}}$.  Then $M \circ \overline M = A$ and $M$ has rank $<m$: as
each column of~$M$ sums to zero, the sum of the $m$ rows is the
0-vector so they are linearly dependent. \lemref{lem:1D} then
completes the proof.
\end{proof}

\subsection{The behavior of PSD-rank under tensoring}

In this subsection, we discuss how PSD-rank behaves under tensoring.
Firstly, we have the following trivial observation on PSD-rank.
\begin{lemma}\label{lem:tensor}
If $P_1$ and $P_2$ are two nonnegative matrices, then it holds that
\[
\prank(P_1\otimes P_2)\leq\prank(P_1) \prank(P_2).
\]
\end{lemma}
\begin{proof}
Suppose $\{C_i\}$ and $\{D_j\}$ form a size-optimal
PSD-factorization of $P_1$, and $\{E_k\}$ and $\{F_l\}$ form a
size-optimal PSD-factorization of $P_2$, where the indices are
determined by the sizes of $P_1$ and $P_2$. Then it can be seen that
$\{C_i\otimes E_k\}$ and $\{D_j\otimes F_l\}$ form a
PSD-factorization of $P_1\otimes P_2$.
\end{proof}

We now consider an example. Let $x,y$ be two subsets of $\{1,2,\ldots,n\}$.
The disjointness function, $\DISJ_n(x,y)$, is defined to be $1$ if
$x\cap y=\varnothing$ and $0$ otherwise.  We denote its corresponding
$2^n$-by-$2^n$ matrix by $D_n$, i.e., $D_n(x,y)=\DISJ_n(x,y)$. This function is one of the most
important and well-studied in communication complexity.  It can be easily checked
that for any natural number $k$, $D_{k}=D_{1}^{\otimes k}$.
According to the above lemma, we have that $\prank(D_n)\leq 2^n$,
where we used the fact that $\prank(D_1)= 2$. This upper bound is trivial
as the size of $D_n$ is $2^n$, but in this case it is tight.
The following lemma was also found independently by G\'abor Braun and Sebastian Pokutta \cite{BP14}.

\begin{lemma}\label{lem:zero} Suppose $A$ is an
$m$-by-$n$ nonnegative matrix, and has the following block
expression,
\[
A=
\begin{bmatrix}
B & C \\
D & 0
\end{bmatrix}.
\]
Then $\prank(A)\geq\prank(C)+\prank(D)$.
\end{lemma}
\begin{proof}
Suppose $\{E_1,E_2,\ldots,E_m\}$ and $\{F_1,F_2,\ldots,F_n\}$ form a
size-optimal PSD-factorization of~$A$. Suppose the size of $B$ is
$k$-by-$l$, then $\{E_1,E_2,\ldots,E_{k}\}$ and
$\{F_{l+1},F_{k+2},\ldots,F_n\}$ form a PSD-factorization of~$C$, while
$\{E_{k+1},E_{k+2},\ldots,E_m\}$ and $\{F_{1},F_{2},\ldots,F_l\}$ form a
PSD-factorization of $D$. According to the definition of
PSD-factorization, the dimension of the support of
$\sum_{i=l+1}^nF_i$ will be at least $\prank(C)$, and similarly, the
dimension of the support of $\sum_{i=k+1}^mE_i$ will be at least
$\prank(D)$.

On the other hand, for any $i\in\{k+1,k+2,\ldots,m\}$ and
$j\in\{l+1,\ldots,n\}$, $\Tr(E_iF_j)=0$, so the support of
$\sum_{i=k+1}^mE_i$ is orthogonal to that of $\sum_{i=l+1}^nF_i$.
Hence $\prank(A)\geq\prank(C)+\prank(D)$.
\end{proof}

Then we have that
\begin{theorem}
$\prank(D_n) = 2^n$.
\end{theorem}
\begin{proof}
Note that for any integer $k$, $D_{k+1}$ can be expressed as the following block matrix.
\[
D_{k+1}=
\begin{bmatrix}
D_k & D_k \\
D_k & 0
\end{bmatrix},
\]
Then by \lemref{lem:zero} we have that
$\prank(D_{k+1})\geq 2\prank(D_k)$. Since $\prank(D_1)=2$, it follows
that $\prank(D_n)\geq 2^n$. Since $\prank(D_n)\leq 2^n$, this completes the proof.
\end{proof}

Based on this example and by analogy to the normal rank, one might conjecture
that generally $\prank(P_1\otimes P_2)=\prank(P_1)\prank(P_2)$.
This is false, however, as shown by the following counterexample.
\begin{example}
Let $A=
\begin{bmatrix}
1 & a \\
a & 1
\end{bmatrix}$ for nonnegative $a$.  Then $A$ has rank~2, and therefore PSD-rank~2, as long as
$a \not = 1$.  On the other hand,
\[
A\otimes A=
\begin{bmatrix}
1 & a & a & a^2\\
a & 1 & a^2 & a\\
a & a^2 & 1 & a\\
a^2 & a & a & 1
\end{bmatrix}
\]
satisfies the condition of \thmref{thm:not_full} for any $a \in [-1+\sqrt{2}, 1+\sqrt{2}]$.
Thus for $a \in [-1+\sqrt{2}, 1+\sqrt{2}] \setminus \{1\}$ we have $\prank(A \otimes A) < \prank(A)^2$.
\end{example}

\subsection{PSD-rank and real PSD-rank}

In the original definition of PSD-rank, the matrices of the
PSD-factorization can be arbitrary complex Hermitian PSD matrices. A
natural and interesting question is what happens if we restrict
these matrices instead to be positive semidefinite \emph{real}
matrices.\footnote{This question was raised by Dirk Oliver Theis in
the Dagstuhl seminar 13082 (February 2013).} We call this restriction the \emph{real PSD-rank},
and for a nonnegative matrix~$A$ we denote it by $\rprank(A)$.
The following observation (proved in the appendix)
shows that the multiplicative gap between these notions cannot be too large.

\begin{theorem}\label{thm:realvscomplex}
If $A$ is a nonnegative matrix, then $\prank(A)\leq\rprank(A)\leq2\prank(A)$.
\end{theorem}

Below in \exampleref{realgap} we will exhibit a gap between $\prank(A)$ and $\rprank(A)$ by a factor of~$\sqrt{2}$.

\section{Three new lower bounds for PSD-rank}
In this section we give three new lower bounds on the PSD-rank. All
of these bounds are based on the interpretation of PSD-rank in terms of
communication complexity.

\subsection{A physical explanation of PSD-rank}
For a nonnegative $m \times n$ matrix $P = [P(i,j)]_{i,j}$, suppose
$\prank(P)=r$. Then there exist $r\times r$ positive
semidefinite matrices $E_i$, $F_j$, satisfying that $P(i,j) = \Tr(E_i  F_j)$,
for every $i\in[m]$ and $j\in[n]$. Fiorini et al.\ show
how from a size-$r$ PSD-factorization of a matrix $P$, one can
construct a one-way quantum communication protocol sending
$(r+1)$-dimensional messages that
computes $P$ in expectation~\cite{FMP+12}.
We will now show that without loss of generality that factors $E_1,
\ldots, E_m, F_1, \ldots, F_n$ have a very particular form. Namely,
we can assume that $\sum_i E_i =I$ (so they form a POVM) and $\Tr(F_j)=1$
(so the $F_j$ can be viewed as quantum states).
We now give a direct proof of this without increasing the size.
This observation will be the key to our lower bounds.
\begin{lemma}\label{lemma:physical}
Let $P$ be an $m$-by-$n$ matrix where each column is a probability
distribution.  If $\prank(P)=r$, then there exists a
PSD-factorization for $P(i,j)=\Tr(E_iF_j)$ such that $\Tr(F_j)=1$
for each~$j$ and
\begin{align*}
\sum_{i=1}^{m}E_i=I,
\end{align*}
where $I$ is the $r$-dimensional identity.
\end{lemma}
\begin{proof}
Suppose $r$-by-$r$ positive semidefinite matrices $C_1, \ldots, C_m$
and $D_1, \ldots, D_n$ form a PSD-factorization for $P$. Note that
for any $r$-by-$r$ unitary matrix $U$, it holds that
\begin{align*}
\Tr ( C_i  D_j)=\Tr ((UC_iU^\dag)(UD_jU^\dag)).
\end{align*}
Therefore $UC_iU^\dag$ and $UD_jU^\dag$ also form a
PSD-factorization for $P$. In the following, we choose $U$ as the
unitary matrix that makes $C'=UCU^\dag$ diagonal, where
$C=\sum_iC_i$.

We first show that $C$ is full-rank.  Suppose not.  Then, without loss of generality, we may assume
$C'$ is a rank-$(r-1)$ diagonal matrix with the $r^{th}$ diagonal
entry being $0$. Since $C'=\sum_iUC_iU^\dag$, we have that for any
$i\in[m]$, the $r^{th}$ column and the $r^{th}$ row of $UC_iU^\dag$
are all zeros. That is to say, in the PSD-factorization for $P$
formed by $UC_iU^\dag$ and $UD_jU^\dag$, the $r$th dimension has no
contribution, resulting in a smaller PSD-factorization for $P$, which is a
contradiction.

Now that $C'$ is full-rank, one can always find another full-rank
nonnegative diagonal matrix $V$ such that $VC'V^\dag=I$. Let
$E_i=VUC_iU^\dag V^\dag$, and $F_j=V^{-1}UD_jU^\dag (V^{-1})^\dag$.
Then it is not difficult to verify that $E_i$ and $F_j$ form another
PSD-factorization for $P$ with size $r$, satisfying $\sum_iE_i=I$.

Finally note that $\Tr(F_j)=\Tr(F_j I)=\sum_i \Tr(E_i F_j)=1$ as each
column of $P$ sums to one.
\end{proof}

\subsection{A lower bound based on fidelity}

\begin{definition}
For nonnegative stochastic matrix~$P$, define
$$
B_3(P)=\max_q \frac{1}{\sum_{i,j} q_i q_j F(P_i,P_j)^2},
$$
where $P_i$ is the $i^{th}$ column of $P$ and the max is taken over
probability distributions $q=\{q_j\}$.
\end{definition}

\begin{theorem}\label{thm:fidelity}
$\prank(P) \ge B_3(P)$.
\end{theorem}
\begin{proof}
Let $\{E_i\}, \{\rho_j\}$ be a size-optimal PSD-factorization of $P$.
According to \lemref{lemma:physical}, we may assume that $\sum_i E_i
= I$ and $\Tr(\rho_j)=1$ for each~$j$.
For a probability distribution $\{q_j\}$, let $\rho=\sum_j q_j
\rho_j$. Notice that the dimension of $\rho$ is $\prank(P)$, thus
the \emph{rank} of $\rho$ will be at most $\prank(P)$.  We use the
trace norm bound Eq.~(\ref{eq:trace_norm_bound}) to lower bound the rank of $\rho$ giving
\[
\prank(P) \ge \frac{\trn{\rho}^2}{\fro{\rho}^2} =
\frac{1}{\fro{\rho}^2}.
\]
Let us now proceed to upper bound $\fro{\rho}^2$.  We have
\[
\fro{\rho}^2 = \Tr(\rho^2)=\sum_{i,j} q_i q_j \Tr(\rho_i\rho_j) \le
\sum_{i,j} q_i q_j F(\rho_i,\rho_j)^2,
\]
where we used \factref{ft:fidelity}. As $P_i$ is obtained from measuring
$\rho_i$ with the POVM $\{E_j\}$, according to
\factref{ft:fidelClasQuan} we have that $F(\rho_i,\rho_j)\le
F(P_i,P_j)$, which gives the bound
$\displaystyle \prank(P) \ge \max_q \frac{1}{\sum_{i,j} q_i q_j F(P_i,P_j)^2}$.
\end{proof}

We can extend the notation $B_3(P)$ to nonnegative matrices $P$ that are not
stochastic, by first normalizing the columns of $P$ to make
it stochastic and then applying $B_3$ to the resulting stochastic matrix.
As rescaling a nonnegative matrix by multiplying its
rows or columns with nonnegative numbers does not increase its
PSD-rank, we have the following definition and corollary.

\begin{definition}
For a nonnegative $m \times n$ matrix $P = [P(i,j)]_{i,j}$, define
$$
B_3'(P)=\max_{q,D} \frac{1}{\sum_{i,j} q_i q_j F((DP)_i,(DP)_j)^2},
$$
where $q=\{q_j\}$ is a probability distribution, $D$ is a diagonal
nonnegative matrix, and $(DP)_i$ is the probability distribution
obtained by normalizing the $i^{th}$ column of $DP$ via a constant factor.
\end{definition}

\begin{corollary}\label{cor:B1rescal}
$\prank(P) \ge B_3'(P)$.
\end{corollary}

We now see an example where rescaling can improve the bound.
\begin{example}
\label{ex:Rescaling} Consider the following $n\times n$ nonnegative
matrix $A$, where $n=10$, and $\epsilon=0.01$.
\[ A = \begin{bmatrix}
1 & 1 & 1 & \cdots & 1 & 1\\
\epsilon & 1 & \epsilon & \cdots & \epsilon & \epsilon\\
\epsilon & \epsilon & 1 & \cdots & \epsilon & \epsilon\\
\vdots & \vdots & \vdots & \ddots  & \vdots & \vdots\\
\epsilon & \epsilon & \epsilon & \cdots & 1 & \epsilon\\
\epsilon & \epsilon & \epsilon & \cdots & \epsilon & 1
\end{bmatrix}. \]
Suppose $P$ is the nonnegative stochastic matrix obtained by
normalizing the columns of $A$ by constant factors, then it has the
same PSD-rank as $A$. By choosing $q$ as the uniform probability
distribution, we can get a lower bound of $B_3(P)$ as follows. Note
that for any $i\in[n]\setminus\{1\}$, we have that
\[
f_1:=F(P_1,P_i)=\frac{1+\sqrt{\epsilon}+(n-2)\epsilon}{\sqrt{1+(n-1)\epsilon}\cdot\sqrt{2+(n-2)\epsilon}},
\]
and for any distinct $i,j\in[n]\setminus\{1\}$, it holds that
\[
f_2:=F(P_i,P_j)=\frac{1+2\sqrt{\epsilon}+(n-3)\epsilon}{2+(n-2)\epsilon}.
\]
Then we get
\[
B_3(A)\geq \frac{n^2}{n+2(n-1)\cdot
{f_1}^2+(n-2)(n-1)\cdot{f_2}^2}\approx 2.09.
\]
We now multiply every row of $A$ by $10$ except that the first one
is multiplied by $0$, i.e., the matrix~$D$ in \corref{cor:B1rescal}
is a diagonal nonnegative matrix with diagonal $(0,10,\ldots,10)$.
Then we obtain another nonnegative matrix $\hat A=DA$. By a similar
calculation as above, it can be verified that $B_3(\hat A)\geq4.88$,
hence we have $B_3'(A)\geq 4.88$, which is a better lower bound.
\end{example}

\subsection{A lower bound based on the structure of POVMs}

\begin{definition}
For nonnegative stochastic matrix~$P$, define $B_4(P)=\sum_i \max_j P(i,j)$.
\end{definition}

\begin{theorem}
$\prank(P) \ge B_4(P).$
\end{theorem}

\begin{proof}
Let $\{E_i\}, \{\rho_j\}$ be a size-optimal PSD-factorization of $P$ with $\sum_i E_i=I$ and
$\Tr(\rho_j)=1$ for each~$j$.  Note that this condition on the trace of $\rho_j$ implies $I \succeq \rho_j$.
Thus
\[
\Tr(E_i)=\Tr(E_i\cdot I)\geq \max_j \Tr(E_i \rho_j)=\max_j P(i,j).
\]
On the other hand, since $\sum_i E_i = I$, we have
\[
\prank(P)=\sum_i\Tr(E_i)\geq\sum_i\max_j P(i,j),
\]
where we used that the size of $I$ is $\prank(P)$.
\end{proof}

A variant of $B_4$ involving rescaling can sometimes lead to better bounds:

\begin{definition}
For a nonnegative $m \times n$ matrix $P = [P(i,j)]_{i,j}$, define
$$
B_4'(P)=\max_D\sum_i \max_j ((DP)_j)_i,
$$
where $D$ is a diagonal nonnegative matrix, $(DP)_j$ is the
probability distribution obtained by normalizing the $j^{th}$ column
of $DP$ via a constant factor, and $((DP)_j)_i$ is the $i^{th}$
entry of $(DP)_j$.
\end{definition}

\begin{corollary}
$\prank(P) \ge B_4'(P)$.
\end{corollary}

\begin{example}
We consider the same matrices $A$ and $D$ as in \exampleref{ex:Rescaling}, and get that
\[
B_4(A)= \frac{1}{1+(n-1)\epsilon}+(n-1)\cdot\frac{1}{2+(n-2)\epsilon}\approx 5.24.
\]
Similarly, it can be checked that $B_4'(A)\geq 8.33$. The latter
indicates that $\prank(A)\geq 9$, which is better than the bound $4$
given by $B_1(A)$ or $6$ by $B_2(A)$.
\end{example}

\subsection{Another bound that combines $B_3$ with $B_4$}

Here we will show that $B_4$ can be strengthened
further by combining it with the idea that bounds $\Tr(\sigma^2)$ in~$B_3$,
where $\sigma$ is a quantum state that can be expressed as some
linear combination of $\rho_i$'s.

\begin{definition}
For a nonnegative stochastic matrix $P = [P(i,j)]_{i,j}$, define
$$
B_5(P)=\sum_i \max_{q^{(i)}}\frac{\sum_k
q^{(i)}_kP(i,k)}{\sqrt{\sum_{s,t} q^{(i)}_s q^{(i)}_t
F(P_s,P_t)^2}},
$$
where $P_s$ is the $s^{th}$ column of $P$, and for every $i$,
$q^{(i)}=\{q^{(i)}_k\}$ is a probability distribution.
\end{definition}

\begin{theorem}
$\prank(P) \ge B_5(P)$.
\end{theorem}

\begin{proof}
We define $\{E_i\}$ and $\{\rho_j\}$ as before. For an
arbitrary~$i$, we define $\sigma_i=\sum_kq^{(i)}_k\rho_k$. This is a
valid quantum state.
Since $\Tr(E_i\rho_j)=P(i,j)$, it holds that
$\Tr(E_i\sigma_i)=\sum_kq^{(i)}_kP(i,k)$. The Cauchy-Schwarz
inequality gives $\Tr^2(E_i\sigma_i)\leq \Tr(E_i^2)\Tr(\sigma_i^2)$.
This implies that
\[
\left(\sum_kq^{(i)}_kP(i,k)\right)^2\leq \Tr^2(E_i)\sum_{s,t}
q^{(i)}_s q^{(i)}_t F(P_s,P_t)^2,
\]
where we used the facts that $\Tr(E_i^2)\leq\Tr^2(E_i)$ and
$\Tr(\sigma_i^2)\leq \sum_{s,t} q^{(i)}_s q^{(i)}_t F(P_s,P_t)^2$;
the latter has been proved in \thmref{thm:fidelity}. Therefore,
for any distribution $q^{(i)}$ it holds that
\[
\Tr(E_i)\geq\frac{\sum_kq^{(i)}_kP(i,k)}{\sqrt{\sum_{s,t} q^{(i)}_s
q^{(i)}_t F(P_s,P_t)^2}}.
\]
Substituting this result into the fact that $\sum_i \Tr(E_i) =
\prank(P)$ completes the proof.
\end{proof}

We also have the following corollary that allows rescaling.

\begin{definition}
For a nonnegative $m \times n$ matrix $P = [P(i,j)]_{i,j}$, define
$$
B'_5(P)=\max_D\sum_i \max_{q^{(i)}}\frac{\sum_k
q^{(i)}_k((DP)_k)_i}{\sqrt{\sum_{s,t} q^{(i)}_s q^{(i)}_t
F((DP)_s,(DP)_t)^2}},
$$
where for every $i$, $q^{(i)}=\{q^{(i)}_k\}$ is a probability
distribution, $D$ is a diagonal nonnegative matrix, $(DP)_k$ is the
probability distribution obtained by normalizing the $k^{th}$ column
of $DP$ via a constant factor, and $((DP)_k)_i$ is the $i^{th}$
entry of $(DP)_k$.
\end{definition}

\begin{corollary}
$\prank(P) \ge B'_5(P)$.
\end{corollary}

We now give an example showing that $B_5$ can be better than $B_4$.

\begin{example}\label{ex:b5good}
Consider the following $n\times n$ nonnegative matrix $A$, where
$n=10$, and $\epsilon=0.01$.
\[ A = \begin{bmatrix}
1 & 1 & \epsilon & \cdots & \epsilon & \epsilon\\
\epsilon & 1 & 1 & \cdots & \epsilon & \epsilon\\
\epsilon & \epsilon & 1 & \cdots & \epsilon & \epsilon\\
\vdots & \vdots & \vdots & \ddots  & \vdots & \vdots\\
\epsilon & \epsilon & \epsilon & \cdots & 1 & 1\\
1 & \epsilon & \epsilon & \cdots & \epsilon & 1
\end{bmatrix}. \]
It can be verified that $B_4(A)\approx 4.81$. In order to provide a
lower bound for $B_5(A)$, for any $i$ we choose $q^{(i)}$ as
$\{0,\ldots,0,1/2,1/2,0,\ldots0\}$, where the positions of $1/2$ are
exactly the same as those of $1$ in the $i^{th}$ row of $A$.
Straightforward calculation shows that $B_5(A)\geq 5.36$, which is
better than $B_4(A)$.
\end{example}

Even $B_5$ can be quite weak in some cases. For example for the
matrix in \exampleref{ex:bad} one can show $B_5(A)<1.1$, which is
weaker than $B_1(A)\approx 3.16$.

\section{Comparisons between the bounds}

In this section we give explicit examples comparing the three new lower bounds on PSD-rank ($B_3$, $B_4$ and $B_5$) and the two that were already known ($B_1$ and $B_2$).

All our examples will only use positive entries, which trivializes all support-based
lower bound methods, i.e., methods the only look at the pattern of zero
and non-zero entries in the matrix. Note that most lower bounds on nonnegative rank
are in fact support-based (one exception is~\cite{FawziParrilo12}).
Since PSD-rank is always less than or equal to nonnegative rank, the
results obtained in the current paper could also serve as new lower
bounds for nonnegative rank that apply to arbitrary nonnegative
matrices. Serving as lower bounds
for nonnegative rank, our bounds are more coarse than the bounds in~\cite{FawziParrilo12}
(this is natural, as we focus on PSD-rank essentially,
and the gap between PSD-rank and nonnegative rank can be very
large~\cite{FMP+12}).  On the other hand, our bounds are much easier to calculate.


The first example indicates that in some cases $B_4$ can be at least
quadratically better than each of $B_1$, $B_2$ and~$B_3$.

\begin{example}
\label{ex:B4Best} Consider the following $(n+1)\times(n+1)$
nonnegative matrix $A$, where $\epsilon=1/n$.
\[ A = \begin{bmatrix}
1 & \epsilon & \epsilon & \cdots & \epsilon & \epsilon\\
\epsilon & 1 & \epsilon & \cdots & \epsilon & \epsilon\\
\epsilon & \epsilon & 1 & \cdots & \epsilon & \epsilon\\
\vdots & \vdots & \vdots & \ddots  & \vdots & \vdots\\
\epsilon & \epsilon & \epsilon & \cdots & 1 & \epsilon\\
\epsilon & \epsilon & \epsilon & \cdots & \epsilon & 1
\end{bmatrix}. \]
\thmref{thm:LowerBoundforApproxId} (below) shows that
$B_4(A)=\frac{n+1}{2}$, and by straightforward calculation one can
also get that $B_1(A)=\sqrt{n}$,
$B_2(A)=\frac{n+1}{2\sqrt{n}}\approx\frac{\sqrt{n}}{2}$, and
numerical calculation indicates that $B_3(A)$ is around $4$.
\end{example}

The second example shows that $B_3$ can also be the best among the
four lower bounds $B_1,B_2,B_3,B_4$, indicating that $B_3$ and $B_4$ are
incomparable.

\begin{example}
\label{ex:B3Best} Consider the following $n\times n$ nonnegative
matrix $A$, where $n=10$, and $\epsilon=0.001$.
\[ A = \begin{bmatrix}
1 & 1 & \epsilon & \cdots & \epsilon & \epsilon\\
1 & 1 & 1 & \cdots & \epsilon & \epsilon\\
\epsilon & 1 & 1 & \cdots & \epsilon & \epsilon\\
\vdots & \vdots & \vdots & \ddots  & \vdots & \vdots\\
\epsilon & \epsilon & \epsilon & \cdots & 1 & 1\\
\epsilon & \epsilon & \epsilon & \cdots & 1 & 1
\end{bmatrix}. \]
That is, $A=(1-\epsilon)\cdot B+\epsilon\cdot J$, where $B$ is the
tridiagonal matrix with all nonzero elements being~1, and $J$ is the
all-one matrix. By straightforward calculation, we find that
$B_1(A)\approx 3.16$, $B_2(A)\approx 3.42$, $B_4(A)\approx 3.99$,
and the calculation based on uniform probability distribution $q$
shows that $B_3(A)\geq 4.52$. The result of $B_3(A)$ shows that
$\prank(A)\geq 5$.
\end{example}

Unfortunately, sometimes $B_3$ and $B_4$ can be very weak bounds%
\footnote{Even though a nonnegative matrix has the same
PSD-rank as its transposition, the bounds given by $B_3$ (or $B_4$)
can be quite different, for instance for the matrix~$A$ of~\exampleref{ex:Rescaling}.}%
, and even the trivial rank-based bound~$B_1$ can be much better than both of them.

\begin{example}\label{ex:bad}
Consider the following $n\times n$ nonnegative matrix $A$, where
$n=10$, and $\epsilon=0.9$.
\[ A = \begin{bmatrix}
1 & \epsilon & \epsilon & \cdots & \epsilon & \epsilon\\
\epsilon & 1 & \epsilon & \cdots & \epsilon & \epsilon\\
\epsilon & \epsilon & 1 & \cdots & \epsilon & \epsilon\\
\vdots & \vdots & \vdots & \ddots  & \vdots & \vdots\\
\epsilon & \epsilon & \epsilon & \cdots & 1 & \epsilon\\
\epsilon & \epsilon & \epsilon & \cdots & \epsilon & 1
\end{bmatrix}. \]
It can be verified that $B_2(A)\approx 1.0005$, and
$B_4(A)\approx1.099$. For $B_3(A)$, numerical calculation indicates
that it is also around~$1$. However, it is easy to see that
$B_1(A)\approx 3.16$. Thus, the best lower bound is given by
$B_1(A)$, i.e., $\prank(A)\geq 4$.
\end{example}

\begin{example}\label{ex:polygon}
For slack matrices of regular polygons, the two new bounds $B_3$ and
$B_4$ are not good either, and in many cases they are at most~3.
Moreover, numerical calculations show that rescaling probably cannot
improve much. Note that the two trivial bounds $B_1$ and $B_2$ are
also very weak for these cases. As an instance, consider a slack
matrix of the regular hexagon \cite{GPT11}
\[ A = \begin{bmatrix}
0 & 0 & 1 & 2 & 2 & 1\\
1 & 0 & 0 & 1 & 2 & 2\\
2 & 1 & 0 & 0 & 1 & 2\\
2 & 2 & 1 & 0  & 0 & 1\\
1 & 2 & 2 & 1  & 0 & 0\\
0 & 1 & 2 & 2 & 1 & 0
\end{bmatrix}. \]
It can be verified that $B_1(A)\approx 1.73$, $B_2(A)\approx 1.59$,
$B_4(A)=6\times \frac{2}{6}=2$, and choosing $q$ in the definition
of $B_3(A)$ as uniform distribution gives that $B_3(A)>2.1$.
Furthermore, our numerical calculations showed that choosing other
distributions or utilizing rescaling could not improve the results
much, and never gave lower bounds $\geq3$.
\end{example}

\section{PSD-factorizations for specific functions}\label{sec:B1tight}
In this section we show the surprising power of PSD-factorizations by giving
nontrivial upper bounds on the PSD-rank of the nonequality and inner product functions.
These bounds are tight up to constant factors.

\subsection{The nonequality function}
The nonequality function defines an $n$-by-$n$ matrix $A_n$ with entries
$A_n(i,i)=0$ and $A_n(i,j)=1$ if $i \ne j$.  In other words, $A_n = J_n - I_n$ where
$J_n$ is the all-ones matrix and $I_n$ is the identity of size~$n$.
This is also known as the ``derangement matrix.'' Note that for $n>1$ it has full rank.

The basic idea of our PSD factorization is the following.  We first construct $n^2$ Hermitian matrices
$G_{ij}$ of size $n$ with spectral norm at most $1$.  Then the matrices
$I + G_{ij}$ and $I-G_{ij}$ will be positive semidefinite, and these
will form the factorization.  Note that
\[
\Tr((I+G_{ij})(I-G_{kl})^*) = \Tr(I) +\Tr(G_{ij}) - \Tr(G_{kl}^*) - \Tr(G_{ij}G_{kl}^*).
\]
Thus if we can design the $G_{ij}$ such that $\Tr(G_{ij}) =
\Tr(G_{kl})$ for all $i,j,k,l$ and $\Tr(G_{ij}G_{kl}) = \delta_{ik}
\delta_{jl} n$ (where $\delta_{ij}=1$ if $i=j$, and $\delta_{ij}=0$ otherwise),
this will give a factorization proportional to the nonequality matrix.

For the case where $n$ is odd, we are able to carry out this plan exactly.
\begin{lemma}\label{lem:odd}
Let $n$ be odd.  Then there are $n^2$ Hermitian matrices $G_{ij}$ of size $n$
such that
\begin{itemize}
  \item $\Tr(G_{ij})=\Tr(G_{kl})$ for all $i,j,k,l\in [n]:=\{0,\ldots,n-1\}$.
  \item $\Tr(G_{ij} G_{kl}^*)=\delta_{ik} \delta_{jl} n$.
  \item $G_{ij} G_{ij}^*=I_n$.
\end{itemize}
\end{lemma}

\begin{proof}
We will use two auxiliary matrices in our construction.  We will label matrix entries from $[n]$.
Let $L$ be the addition table of $\mathbb{Z}_n$, that is $L(i,j)=i+j
\mod n$.  Notice that $L$ is a symmetric Latin square\footnote{A Latin square is
an $n$-by-$n$ matrix in which each row and each column is a permutation of $[n]$.}
with distinct entries along the main diagonal.  Let $V$ be the Vandermonde matrix
that is $V(k,l)=e^{-2kl\pi\rm{i} / n}$ for $k, l \in [n]$.  Note that
$VV^*=nI_n$.

We now define the matrices $G_{ij}$ for $i,j \in [n]$.  The matrix
$G_{ij}$ will be nonzero only in those entries where $L(k,l)=i$.
Thus the zero/nonzero pattern of each $G_{ij}$ forms a permutation
matrix with exactly one~1 on the diagonal.  These nonzero entries
will be filled in from the $j$th row of $V$.  We do this in a way to
ensure that $G_{ij}$ is Hermitian.  Thus $V(j,0)=1$ will be placed on
the diagonal entry of $G_{ij}$.  Now fix an ordering of the
$\floor{n/2}$ other pairs $(k,l),(l,k)$ of nonzero entries of
$G_{ij}$ (say that each $(k,l)$ is above the diagonal).  In the
$t^{th}$ such pair we put the conjugate pair $V(j,t), V(j,n-t)$.  In
this way, $G_{ij}$ is Hermitian, and as the ordering is the same for
all $j$ we have that $\Tr(G_{ij} G_{ik}^*) = \braket{V_i}{V_k}=n
\delta_{i,k}$.

To finish, we check the other properties.  Each $G_{ij}$ has trace
one. If $i \ne k$ then $\Tr(G_{ij} G_{kl}^*)=0$ as the zero/nonzero
patterns are disjoint.  Finally as the zero/nonzero pattern of each
$G_{ij}$ is a permutation matrix, and entries are roots of unity,
$G_{ij} G_{ij}^*=I_n$.
\end{proof}

This gives the following theorem for the $n^2$-by-$n^2$ nonequality matrix.

\begin{theorem}\label{thm:NEOdd}
Suppose $n$ is odd, and let $A_{n^2}$ be nonequality matrix of size $n^2$.  Then $\prank(A_{n^2}) \leq n$.
\end{theorem}

\begin{proof}
Suppose $n^2$ Hermitian matrices $G_{ij}$ have been constructed as
in \lemref{lem:odd}. We now define the matrices
$X_{ij}=(1/\sqrt{n})(I + G_{ij})$ and
$Y_{ij}=(1/\sqrt{n})(I-G_{ij}^*)$. Note that the spectral norm of
each $G_{ij}$ is~1, so $X_{ij}$ and $Y_{ij}$ are PSD.
Also, we have
\begin{align*}
\Tr(X_{ij} Y_{kl}) &= \frac{1}{n} \left( \Tr(I) +\Tr(G_{ij}) - \Tr(G_{kl}^*) - \Tr(G_{ij} G_{kl}^*) \right) \\
&= \frac{1}{n} \left ( n - \delta_{ik} \delta_{jl} n \right) = 1 - \delta_{ik} \delta_{jl}.
\end{align*}
\vspace*{-1em}
\end{proof}

We now turn to the case that $n$ is even. The result is slightly worse here.
\begin{lemma}\label{lem:NEEven}
Let $n$ be even.  Then there are $n^2-1$ Hermitian matrices $G_{ij}$
such that
\begin{itemize}
  \item $\Tr(G_{ij})=\Tr(G_{kl})$ for all $i,j,k,l$.
  \item $\Tr(G_{ij} G_{kl}^*)=\delta_{ik} \delta_{jl} n$.
  \item $G_{ij} G_{ij}^*=I_n$.
\end{itemize}
\end{lemma}
\begin{proof}
The construction is similar.  Again let $V$ be the Vandermonde
matrix of roots of unity and this time let $L$ be a symmetric Latin
square with entries from $[n]$ where the diagonal has all entries
$0$.

For $i>0$, the matrix $G_{ij}$ is defined as before, with the additional
subtlety that if $j$ is odd then $V(j,0)=1$ and $V(j,n/2)=-1$ and
instead of taking this pair we use $(\rm{i},-\rm{i})$ in the
matrix.

For $i=0$ we use all the rows of $V$ except $V_0$,  the all-one row,
to ensure that the trace of all $G_{ij}$ is zero (this is why we can only
create $n^2-1$ matrices).
\end{proof}

As with the case where $n$ is odd, we have the following theorem
based on \lemref{lem:NEEven}.
\begin{theorem}\label{thm:NEEven}
Suppose $n$ is even, and let $A_{n^2-1}$ be the nonequality matrix
of size $n^2-1$. Then it holds that $\prank(A_{n^2-1}) \leq n$.
\end{theorem}

\noindent
The nonequality function gives a family of matrices where PSD-rank is smaller than the real PSD-rank.

\begin{example}
\label{realgap} We have seen that for odd~$n$, the PSD-rank of the
nonequality matrix of size $n^2$ is at most~$n$.  This is tight by
\factref{ft:trivial}, since the rank of the nonequality matrix of
this size is $n^2$. On the other hand, also by \factref{ft:trivial},
the real PSD-rank is at least $\ceil{\sqrt{2}n-1/2}$, and actually this
bound has been shown to be tight~\cite[Example~5.1]{FGPRT14}. This shows a
multiplicative gap of approximately $\sqrt{2}$ between the real and
complex PSD-rank.

Fawzi et al.~\cite[Section~2.2]{FGPRT14} independently observed that the real and
complex PSD-rank are not the same, showing that the $4$-by-$4$
derangement matrix has complex PSD-rank~$2$, while by \factref{ft:trivial}
the real PSD-rank is at least $3$.  
\end{example}





It should be pointed out that the results in the current subsection
reveal a fundamental difference between PSD-rank and the normal
rank. Recall that for the normal rank we have that $\rank(A-B) \geq
\rank(B) - \rank(A)$. Thus if $A$ is a rank-one matrix, the ranks of
$A-B$ and $B$ cannot be very different. The results above, on the
other hand, indicate that the situation is very different for
PSD-rank, where $A-B$ and $B$ can have vastly different PSD-ranks
even for a rank-one matrix~$A$. This fact shows that the PSD-rank is not
as robust to perturbations as the normal rank, a contributing
reason to why the PSD-rank is difficult to bound.

\begin{proposition}
For every positive integer $d$, there exists a nonnegative
matrix $A$, such that $J-A$ is also nonnegative, and
\[
\left|\prank(J-A)-\prank(A)\right|>d,
\]
where $J$ is the all-one matrix.
\end{proposition}
\begin{proof}
Choose $A=I$, and the size to be $n$, then we have that
$\prank(J-A)\approx\sqrt{n}$, while $\prank(A)=n$.  Choosing
$n$ large enough gives the desired separation.
\end{proof}

\subsection{Approximations of the identity}
Here we first consider the PSD-rank of \emph{approximations} of
the identity.  We say that an $n$-by-$n$ matrix $A$ is an
$\epsilon$-approximation of the identity if $A(i,i)=1$ for all $i
\in [n]$ and $0 \le A(i,j) \le \epsilon$ for all $i \ne j$. The
usual rank of approximations of the identity has been well studied
\cite{Alo09}.

In particular, it is easy to show that if $A$ is an
$\epsilon$-approximation of the identity then
\[
\rank(A) \ge \frac{n}{1+\epsilon^2(n-1)}.
\]
Using the bound $B_4$ we can show a very analogous result for PSD-rank.

\begin{theorem}\label{thm:LowerBoundforApproxId}
If an $n$-by-$n$ matrix~$A$ is an $\epsilon$-approximation
of the identity, then
\[
\prank(A) \ge \frac{n}{1+ \epsilon (n-1)}.
\]
In particular, if $\epsilon \le 1/n$ then $\prank(A) > n/2$.
\end{theorem}

\begin{proof}
We first normalize each column of $A$ to a probability distribution,
obtaining a stochastic matrix $P$.  Each column will be divided by a
number at most $1+\epsilon(n-1)$.  Thus the largest entry of each
column is at least $1/(1+\epsilon(n-1))$.  Hence the method $B_4$
gives the claimed bound.
\end{proof}

%

We now show that this bound is tight in the case of small
$\epsilon$.  If $\epsilon \ge 1/(n-1)^2$, then by \thmref{thm:not_full} the PSD-rank
of the $n$-by-$n$ matrix with ones on the diagonal and $\epsilon$ off the diagonal
is not full.  On the other hand, if $\epsilon < 1/(n-1)^2$ then any $\epsilon$-approximation
of the identity has full PSD-rank, by \thmref{thm:LowerBoundforApproxId}.  This gives
the following proposition.

\begin{proposition}\label{prop:full}
Suppose $A(i,i)=1$ for all $i \in [n]$ and $A(i,j)=\epsilon$ for $i
\ne j$, then $\prank(A)=n$ if and only if $\epsilon<1/(n-1)^2$.
\end{proposition}
\begin{proposition}\label{prop:divide}
Let $m$ divide $n$ and consider the $m$-by-$m$ matrix $B$ where
$B(i,i)=1$ and $B(i,j)=1/(m-1)^2$.  Then $A=I_{n/m}\otimes B$ is an
$\epsilon$-approximation of the identity, and $\prank(A) \le n-
\frac{n}{m}$, where $\epsilon = 1/(m-1)^2$.
\end{proposition}
\begin{proof}
Consider \lemref{lem:tensor} and the fact that $\prank(B)=m-1$.
\end{proof}

As a generalization of approximations of the identity with the same
off-diagonal entries, we now turn to consider the PSD-rank of the
following class of matrices.
\[ M_c = \begin{bmatrix}
c & 1 & 1 & \cdots & 1 & 1\\
1 & c & 1 & \cdots & 1 & 1\\
1 & 1 & c & \cdots & 1 & 1\\
\vdots & \vdots & \vdots & \ddots  & \vdots & \vdots\\
1 & 1 & 1 & \cdots & c & 1\\
1 & 1 & 1 & \cdots & 1 & c
\end{bmatrix}, \]
where $c$ could be any nonnegative real number, and suppose the size
of $M_c$ is $n$-by-$n$. For $c=0$, $M_c$ is
exactly the matrix corresponding to the Nonequality function.
Besides, if $c>(n-1)^2$, \propref{prop:full} implies that the
PSD-rank of $M_c$ will be full. In both of these two cases, our
results are very tight. Then a natural question is, how about the
case when $0<c<(n-1)^2$ (excluding $c=1$)? For this case, it turns out that we
have the following theorem. Combined with $B_1(M_c)=\sqrt{n}$, this
result indicates that when $c$ is not very large, $\prank(M_c)$ is
very small, which is much stronger than \propref{prop:divide}.
\begin{theorem}\label{thm:PsdRankOfMc}
If $c>2$, $\rprank(M_c) \leq 2\lceil
c\rceil\cdot\lceil\sqrt{n}\rceil$. If $c\in[0,2]$, $\rprank(M_c)
\leq \lceil\sqrt{2n}\rceil+1$.
\end{theorem}
\begin{proof}
We first suppose $c$ is an integer larger than $2$. For a fixed
$r\geq c$, we consider the largest set $\mathcal{S}$ of subsets of
$[r]$ such that every subset has exactly $c$ elements and the
intersection of any two subsets contains at most one element in
$[r]$. Suppose the cardinality of $\mathcal{S}$ is $p(r,c)$, and the
elements of $\mathcal{S}$ are $\{S_1,S_2,\ldots,S_{p(r,c)}\}$, i.e.,
for any $i\in[p(r,c)]$, $S_i$ is a subset of $[r]$ with size $c$.

For any $i\in[p(r,c)]$, we now construct two $r$-by-$r$ matrices
$E_i$ and $F_i$ based on $S_i$ as follows. In $E_i$, we first choose
the submatrix whose row index set and column index set are $S_i$,
and let this submatrix be a $c$-by-$c$ all-one
matrix. All the other entries of $E_i$ are set to $0$.
$F_i$ is similar to $E_i$ except that all its diagonal
entries are $1$. Thus, for every $i$, both $E_i$ and $F_i$
are positive semidefinite.

It is not difficult to verify that for any $x,y\in[p(r,c)]$, if
$x=y$ then $\Tr(E_xF_y)=c^2$, and if $x\neq y$ then $\Tr(E_xF_y)=c$. That is,
if $p(r,c)\geq n$, then $\{\frac{1}{c}E_1,\ldots,\frac{1}{c}E_n\}$ and $\{F_1,\ldots,F_n\}$
form a size-$r$ PSD-factorization of $M_c$, which shows that
$\rprank(M_c)\leq r$.
We have the following lemma to provide bounds on $p(r,c)$.
\begin{lemma}
Let $c$ be a positive integer and $q$ be a prime number. There
exists a family of $q^2$ $c$-element sets over a universe of size
$cq$, such that any two distinct sets from this family intersect in
at most one point.
\end{lemma}

\begin{proof}
Since $q$ is a prime number, $\mathbb{F}_q$ is a finite field. With
each $(a,b)\in\mathbb{F}_q\times\mathbb{F}_q$ we associate the
following set in the universe $[c]\times\mathbb{F}_q$. It is a
$c$-element subset of the graph of the line $y=ax+b$.
$$
S_{ab}=\{(x,ax+b) : x\in[c]\}.
$$
We have $q^2$ such sets, one for each choice of $a,b$. Since two
distinct lines can intersect in at most one $(x,y)$-pair, we have
$|S_{ab}\cap S_{a'b'}|\leq 1$ if $a\neq a'$ and/or $b\neq b'$.
\end{proof}

Let us go back to the proof for \thmref{thm:PsdRankOfMc}.
Let $q$ be the smallest prime number $\geq \lceil\sqrt{n}\rceil$, then we know $q\leq 2\lceil\sqrt{n}\rceil$. Now by the above lemma there exist
$q^2\geq n$ $c$-element sets over a universe of size $cq$. This results
in a PSD-factorization for $M_c$ of size~$cq$, hence $\rprank(M_c)\leq cq\leq 2c\cdot\lceil\sqrt{n}\rceil$.

We now turn to the case that $c>2$ and $c$ is not an integer.
Firstly, we construct the PSD-factorization for $M_{\lceil c\rceil}$
as above. Then we replace all the nonzero off-diagonal entries of
the $E_i$'s (which are $1$'s) by $a=\frac{c-1}{\lceil c\rceil-1}$, and
obtain $E'_i$'s. Now $\{E'_1,\ldots,E'_n\}$ and
$\{F_1,\ldots,F_{n}\}$ form a PSD-factorization for $M_c$.

Finally, in order to settle the case that $c\in[0,2]$, we first
focus on the special case that $c=2$. It is easy to see that in this
case, $p(r,c)=\frac{1}{2}r(r-1)$. Thus if we choose
$r=\lceil\sqrt{2n}\rceil+1$, it holds that $p(r,c)\geq n$, and we
have $\rprank(M_2)\leq\lceil\sqrt{2n}\rceil+1$. When
$c\in[0,2)$, we replace all the nonzero off-diagonal entries of
the $E_i$'s (which are $1$'s) by $c-1$, and obtain $E'_i$'s.
It can be verified that $\{E'_1,\ldots,E'_n\}$ and
$\{F_1,\ldots,F_{n}\}$ form a valid PSD-factorization for $M_c$.
\end{proof}

We now consider a more general approximation of the identity than
$M_c$, where the diagonal entries do not have to be 1, and the
off-diagonal entries do not have to be equal. Alon~\cite{Alo09}
proved:

\begin{theorem}(\cite{Alo09})\label{thm:Alon}
There exists an absolute positive constant $c$ so that the following
holds. Let $A=[a(i,j)]$ be an $n$-by-$n$ real matrix with
$|a(i,i)|\geq1/2$ for all $i$ and $|a(i,j)|\leq \epsilon$ for any
$i\neq j$, where $\frac{1}{2\sqrt{n}}\leq\epsilon\leq 1/4$. Then the
rank of $A$ satisfies
\[
\rank(A)\geq\frac{c\log{n}}{\epsilon^2\log{(1/\epsilon)}}.
\]
\end{theorem}

Combining the above theorem and \factref{ft:trivial}, we immediately obtain that
\begin{theorem}\label{thm:PSDOnAlon}
There exists an absolute positive constant $c$ so that the following
holds. Let $A=[a(i,j)]$ be an $n$-by-$n$ real matrix with
$|a(i,i)|\geq1/2$ for all $i$ and $|a(i,j)|\leq \epsilon$ for any
$i\neq j$, where $\frac{1}{2\sqrt{n}}\leq\epsilon\leq 1/4$. Then the
PSD-rank of $A$ satisfies
\[
\prank(A)\geq\frac{c\sqrt{\log{n}}}{\epsilon\sqrt{\log{(1/\epsilon)}}}.
\]
\end{theorem}

We do not know if this lower bound on PSD-rank is tight. It is not hard to show that the \emph{nonnegative} rank of approximations of the $n$-by-$n$ identity matrix is $O(\log n)$ for constant~$\eps$.  For example, we can take a set of $n$ random $\ell$-bit words $C_1,\ldots,C_n\in\01^\ell$. For $\ell=c\log n$ and $c$ a sufficiently large constant, $\braket{C_i}{C_j}$ will be close to $\ell/2$ for all $i=j$ and close to $\ell/4$ for all $i\neq j$. Hence if we associate both the $i$th row and the $i$th column with the $\ell$-dimension vector $\sqrt{\frac{2}{\ell}}C_i$, we get an $\ell=O(\log n)$-dimensional nonnegative factorization of an approximation of the identity.

\subsection{The inner product function}

Let $x,y\in\01^n$ be two $n$-bit strings. The inner product
function is defined as $\IP(x,y)=\sum^n_{i=1}x_iy_i \bmod 2$. We denote the
corresponding $N$-by-$N$ matrix by $\IP_n$, where
$N=2^n$. We have the following theorem.
\begin{theorem}\label{thm:IPExpectProof}
$\prank(\IP_n)\leq 2\sqrt{N}$.
\end{theorem}
\begin{proof}
We will design a one-way quantum protocol to compute $\IP_n$ in expectation
and then invoke the equivalence between $\prank$ and communication
complexity mentioned in \secref{ssecquantum}. We will actually prove the bound for
more general $0/1$-matrices, of which $\IP_n$ is a special case.
Let $W$ be an $N$-by-$N$ $0/1$-matrix,
with rows and columns indexed by $n$-bit strings $x$ and $y$
respectively. View $x=x_0x_1$ as concatenation of two $n/2$-bit
strings $x_0$ and $x_1$. Suppose there exist two Boolean functions
$f,g:\01^{n/2+n}\rightarrow\01$ such that
$W(x,y)=f(x_0,y) \oplus g(x_1,y)$. Then $\IP_n$ is a special case of
such a~$W$, where $f(x_0,y)=\IP(x_0,y_0)$ and $g(x_1,y)=\IP(x_1,y_1)$.
We now show there exists a one-way quantum protocol that computes
$W$ in expectation and whose quantum communication complexity is at
most $n/2+1$ qubits. This implies $\prank(W)\leq 2^{n/2+1}=2\sqrt{N}$.

For any input $x$, \alice sends the following state of $1+n/2$
qubits to \bob:
$$
\ket{\psi_x}=\frac{1}{\sqrt{2}}(\ket{0,x_0}+\ket{1,x_1}).
$$
Then by a unitary operation, \bob turns the state into
\[
\ket{\psi_{xy}}=\frac{1}{\sqrt{2}}((-1)^{f(x_0,y)}\ket{0,x_0}+(-1)^{g(x_1,y)}\ket{1,x_1}).
\]
\bob then applies the Hadamard gate to the last $n/2$ qubits and
measures those in the computational basis. If he gets any
outcome other than $0^{n/2}$, he outputs $0$. With probability $1/\sqrt{2^n}$
 he gets outcome $0^{n/2}$, and then the first qubit will have become
$\frac{1}{\sqrt{2}}((-1)^{f(x_0,y)}\ket{0} + (-1)^{g(x_1,y)}\ket{1})$. By another Hadamard
gate and a measurement in the computational basis, \bob learns the
bit $f(x_0,y)\oplus g(x_1,y)=W(x,y)$. Then he outputs that bit
times $\sqrt{2^n}$. The expected value of the output is
$\frac{1}{\sqrt{2^n}}\cdot(W(x,y)\cdot\sqrt{2^n})=W(x,y)$.
\end{proof}

We give another proof of this theorem by explicitly providing a
PSD-factorization for $\IP_n$.
Note that the factors in the following PSD-factorization are rank-1 real matrices.
\begin{theorem}
$\rprank(\IP_n)\leq c\sqrt{N}$. If $n$ is even, $c=2$, and if $n$ is
odd, $c=\frac{3}{2}\sqrt{2}$.
\end{theorem}
\begin{proof}
For any $k$ we have $\displaystyle
\IP_{k+1}=
\begin{bmatrix}
\IP_k & \IP_k \\
\IP_k & J_k-\IP_k
\end{bmatrix}$,
where $J_k$ is the $k$-by-$k$ all one matrix.
Using this relation twice, we have that
\[
\IP_{k+2}=
\begin{bmatrix}
\IP_k & \IP_k & \IP_k & \IP_k \\
\IP_k & J_k-\IP_k & \IP_k & J_k-\IP_k\\
\IP_k & \IP_k & J_k-\IP_k & J_k-\IP_k \\
\IP_k & J_k-\IP_k & J_k-\IP_k & \IP_k
\end{bmatrix}.
\]
Repeating this procedure, it can be seen that
$\IP_n$ can be expressed as a block matrix with each block being
$\IP_k$ or $J-\IP_k$ for some $k<n$ to be chosen later.
We now consider a new block matrix $M_n$ with
the same block configuration as $\IP_n$ generated as follows. The
blocks in the first block row of $M_n$ are the same as $\IP_n$,
that is they are $\IP_k$'s. In the rest of the block rows, if a block of
$\IP_n$ is $\IP_k$, then we choose the corresponding block of $M_n$ to
be $-\IP_k$, and if a block of $\IP_n$ is $J_k-\IP_k$, the corresponding
block of $M_n$ is also $J_k-\IP_k$. It is not difficult to check that
$M_n\circ \overline M_n = \IP_n$, and since $M_n$ is real, we have that
$\rprank(\IP_n)\leq\rank(M_n)$.

In order to upper bound the rank of $M_n$, we add its first block
row to the other block rows, and obtain another matrix $M'_n$, with
the same rank as $M_n$, in which all the blocks are $0$ or $J_k$
except those in the first row are still $\IP_k$'s. Since the rank of
$M'_n$ can be upper bounded by the sum of the rank of the
first block row and that of the remaining block rows, we have that
\[
\rprank(\IP_n)\leq\rank(M_n)=\rank(M'_n)\leq2^k-1+\frac{N}{2^k},
\]
where $2^k-1$ comes from the rank of $\IP_k$, and $\frac{N}{2^k}$
comes from the number of blocks in every row of $M'_n$. If $n$ is
even, we choose $k=n/2$, and the inequality above is
$\rprank(\IP_n)\leq2\sqrt{N}-1$. If $n$ is odd, we choose
$k=(n+1)/2$, and the inequality becomes
$\rprank(\IP_n)\leq(\frac{3}{2}\sqrt{2})\sqrt{N}-1$.
\end{proof}

\paragraph{Acknowledgments.}
We would like to thank Rahul Jain for helpful discussions, and Hamza Fawzi, Richard Robinson, and Rekha Thomas
for sharing their results on the derangement matrix. Troy Lee and Zhaohui Wei are supported in part by
the Singapore National Research Foundation under NRF RF Award No.~NRF-NRFF2013-13.
Ronald de Wolf is partially supported by ERC Consolidator Grant QPROGRESS and by the EU STREP project QALGO (Grant agreement no.~600700).

\bibliographystyle{alpha}
\bibliography{nnegrk.bib}

\newcommand{\etalchar}[1]{$^{#1}$}
\begin{thebibliography}{FMP{\etalchar{+}}12}

\bibitem[Alo09]{Alo09}
N.~Alon.
\newblock Perturbed identity matrices have high rank: proof and applications.
\newblock {\em Combinatorics, Probability, and Computing}, 18:3--15, 2009.

\bibitem[BP]{BP14}
G.~Braun and S.~Pokutta.
\newblock personal communication.

\bibitem[BTN01]{BentalNemirovski01}
A.~Ben-Tal and A.~Nemirovski.
\newblock On polyhedral approximations of the second-order cone.
\newblock {\em Math. Oper. Res.}, 26(2):193--205, 2001.

\bibitem[CFFT12]{CFF+12}
M.~Conforti, Y.~Faenza, S.~Fiorini, and {H. R.} Tiwary.
\newblock Extended formulations, non-negative factorizations and randomized
  communication protocols.
\newblock In {\em 2nd International Symposium on Combinatorial Optimization},
  2012.
\newblock \href{http://arxiv.org/abs/1105.4127}{arXiv:1105.4127}.

\bibitem[FGP{\etalchar{+}}14]{FGPRT14}
H.~Fawzi, J.~Gouveia, P.~Parrilo, R.~Robinson, and R.~Thomas.
\newblock Positive semidefinite rank.
\newblock Technical Report arXiv:1407.4095, arXiv, 2014.

\bibitem[FMP{\etalchar{+}}12]{FMP+12}
S.~Fiorini, S.~Massar, S.~Pokutta, H.~R. Tiwary, and {R. de} Wolf.
\newblock Linear vs.\ semidefinite extended formulations: Exponential
  separation and strong lower bounds.
\newblock In {\em STOC}, 2012.

\bibitem[FP12]{FawziParrilo12}
H.~Fawzi and P.~Parrilo.
\newblock New lower bounds on nonnegative rank using conic programming.
\newblock Technical Report arXiv:1210.6970, arXiv, 2012.

\bibitem[GPT13]{GPT11}
J.~Gouveia, P.~Parrilo, and R.~Thomas.
\newblock Lifts of convex sets and cone factorizations.
\newblock {\em Mathematics of Operations Research}, 38(2):248--264, 2013.
\newblock \href{http://arxiv.org/abs/1111.3164}{arXiv:1111.3164}.

\bibitem[JSWZ13]{JSWZ12}
R.~Jain, Y.~Shi, Z.~Wei, and S.~Zhang.
\newblock Efficient protocols for generating bipartite classical distributions
  and quantum states.
\newblock {\em IEEE Transactions on Information Theory}, 59:5171--5178, 2013.

\bibitem[LT12]{LeeTheis12}
T.~Lee and D.~O. Theis.
\newblock Support based bounds for positive semidefinite rank.
\newblock Technical Report arXiv:1203.3961, arXiv, 2012.

\bibitem[NC00]{NC00}
M.~Nielsen and I.~Chuang.
\newblock {\em Quantum Computation and Quantum Information}.
\newblock Cambridge University Press, 2000.

\bibitem[Rot14]{Rot14}
T.~Rothvo{\ss}.
\newblock The matching polytope has exponential extension complexity.
\newblock In {\em Proceedings of the 46th ACM STOC}, 2014.

\bibitem[Yan91]{Yannakakis91}
M.~Yannakakis.
\newblock Expressing combinatorial optimization problems by linear programs.
\newblock {\em J. Comput. System Sci.}, 43(3):441--466, 1991.

\bibitem[Zha12]{Zha12}
S.~Zhang.
\newblock Quantum strategic game theory.
\newblock In {\em Proceedings of the 3rd Innovations in Theoretical Computer
  Science}, pages 39--59, 2012.

\end{thebibliography}

%
%
%
%
%
%
%
%

\appendix

\section{Proof of \thmref{thm:realvscomplex}}

It is trivial that $\prank(A)\leq\rprank(A)$, so we only need to
prove the second inequality. Suppose $r=\prank(A)$, and
$\{E_k\}$ and $\{F_l\}$ are a size-optimal PSD-factorization of $A$.
We now separate all the matrices involved into their real and
imaginary parts. Specifically, for any $k$ and $l$, let
$E_k=C_k+i\cdot D_k$, and $F_l=G_l+i\cdot H_l$, where $C_k$ and
$G_l$ are real symmetric matrices, and $D_k$ and $H_l$ are real
skew-symmetric matrices (i.e., $D_k^T=-D_k$ and $H_l^T=-H_l$).
Then it holds that
\[
A_{kl}=\Tr(E_kF_l)=(\Tr(C_kG_l)-\Tr(D_kH_l))+i\cdot(\Tr(D_kG_l)+\Tr(C_kH_l)).
\]
Since $A_{kl}$ is real, we in fact have
\[
A_{kl}=\Tr(C_kG_l)-\Tr(D_kH_l).
\]
Now for any $k$ and $l$, define new matrices as follows:
$S_k=\frac{1}{\sqrt{2}}
\begin{bmatrix}
C_k & D_k \\
-D_k & C_k
\end{bmatrix}$, and $T_l=
\frac{1}{\sqrt{2}}\begin{bmatrix}
G_l & H_l \\
-H_l & G_l
\end{bmatrix}$.
Then $S_k$ and $T_l$ are real symmetric matrices, and $\Tr(S_kT_l)=\Tr(C_kG_l)-\Tr(D_kH_l)=A_{kl}$.

It remains to show that the matrices $S_k$ and $T_l$ are positive semidefinite.
Suppose $u=\begin{bmatrix}
v_1\\
v_2
\end{bmatrix}$
is a $2r$-dimensional real vector, where $v_1$ and $v_2$ are two arbitrary $r$-dimensional real vectors. Starting from the fact that $E_k$ is positive semidefinite, we have
\[
0\leq (v_2^T-i\cdot v_1^T)E_k(v_2+i\cdot
v_1)=v_1^TC_kv_1-v_2^TD_kv_1+v_1^TD_kv_2+v_2^TC_kv_2=\sqrt{2} u^TS_ku.
\]
Hence $S_k$ is positive semidefinite. Similarly we can show that $T_l$ is positive semidefinite for every $l$.

\end{document}